\documentclass[11pt]{article}
\usepackage{graphicx,pict2e,amssymb,amsmath,fullpage,float}
\usepackage{ntheorem,delarray}
\usepackage{hyperref,srcltx}
\usepackage{makeidx}
\makeindex\nonstopmode

\long\def\hide#1{}

\newtheorem{theorem}{Theorem}[section]

\newtheorem{lemma}[theorem]{Lemma}

\newtheorem{claim}[theorem]{Claim}

\newtheorem{definition}{Definition}[section]

\newtheorem{exa}[theorem]{Example}

\newenvironment{proof}{\noindent{\bf Proof.}}{\hfill$\square$\medskip}

\newenvironment{proof*}[1]{\noindent{\bf Proof of #1.}}{\ $\square$\medskip}

\def\FullBox{\hbox{\vrule width 8pt height 8pt depth 0pt}}
\def\qed{\ifmmode\qquad\FullBox\else{\unskip\nobreak\hfil
\penalty50\hskip1em\null\nobreak\hfil\FullBox
\parfillskip=0pt\finalhyphendemerits=0\endgraf}\fi}

\newenvironment{proofof}[1]{\begin{trivlist} \item {\bf Proof
#1:~~}}
  {\qed\end{trivlist}}

\def\Var{\text {Var}}
\def\E{{\sf E}}\def\Var{{\sf Var}}

\def\eps{\varepsilon}

\def\prob{{\sf Prob}}

\def\R{\Re}

\begin{document}

\title{Principal Component Analysis and Higher Correlations for Distributed Data}

\author{Ravindran Kannan\\ Microsoft Research India\\ 
{\tt kannan@microsoft.com}\\ 
\and
Santosh S. Vempala\\ Georgia Tech \\
{\tt vempala@gatech.edu}
\and
David P. Woodruff \\ IBM Research Almaden \\
{\tt dpwoodru@us.ubm.com}}
\date{}

\maketitle

\begin{abstract}
We consider algorithmic problems in the setting in which the input data has been partitioned arbitrarily on many servers. The goal is to compute a function of all the data, and the bottleneck is the communication used by the algorithm. We present algorithms for two illustrative problems on massive data sets: (1) computing a low-rank approximation of a matrix $A=A^1 + A^2 + \ldots + A^s$, with matrix $A^t$ stored on server $t$ and (2) computing a function of a vector $a_1 + a_2 + \ldots + a_s$, where server $t$ has the vector $a_t$; this includes the well-studied special case of computing frequency moments and separable functions, as well as higher-order correlations such as the number of subgraphs of a specified type occurring in a graph. For both problems we give algorithms with nearly optimal communication, and in particular the only dependence on $n$, the size of the data, is in the number of bits needed to represent indices and words ($O(\log n)$). 
%
\end{abstract}


\section{Introduction}
In modern large-scale machine learning problems the input data is often distributed among many servers,
while the communication as well as time and space resources per server are limited. We consider
two well-studied problems:
(1) Principal Component Analysis (PCA), and (2) Generalized Higher-order correlations. 
Both problems study correlations between vectors. For the first problem, the vectors
correspond to the rows of a matrix and we are interested in second-order
correlations, while in the second problem we are interested in higher-order
correlations among the vectors. 

PCA is a central tool in many learning algorithms.
The goal of PCA is to find a low-dimensional subspace that captures as much
of the variance of a dataset as possible. By projecting the rows of a matrix
onto this lower-dimensional subspace, one preserves important properties
of the input matrix, but can now run subsequent algorithms in the 
lower-dimensional space, resulting in significant computational and storage
savings. In a distributed setting, by having each server first locally project
his/her own data onto a low-dimensional subspace, this can also result in
savings in communication. PCA is useful for a variety of downstream tasks, e.g., for
clustering or shape-fitting problems (\cite{fss13}) and latent semantic analysis. 

The second problem we consider is the Generalized Higher Order Correlation Problem. For this problem we assume server $t$ has an
$n$-dimensional vector $a_t$ with non-negative entries. Note that for PCA, it is useful and more general to allow the
entries to be positive,
negative, or zero. On the other hand, 
the non-negativity assumption for Generalized Higher Order Correlations is justified both by the applications we give
, as well as the fact that it is impossible to achieve low communication without this assumption, as described
in more detail below. 

A special case of this problem is the well-studied frequency moment
problem. That is, if server $t$ holds the vector $a_t$, with coordinates $a_{t1}, a_{t2}, \ldots, a_{tn}$, 
then the $k$-th frequency moment of $\sum_{t=1}^s a_t$ is $\sum_{i=1}^n (\sum_{t=1}^s a_{ti})^k$, where, $k$ is a positive integer.
This problem has been extensively studied in
the data stream literature, starting with the work of \cite{AMS99}. 
Known lower bounds for this problem from that literature 
rule out low communication algorithms when $k>2$ in the distributed setting when the number of 
servers grows as a power of $n$ (\cite{Bar-Yossef04, Chakrabarti03,g09}), 
or when there are only two servers and the entries are allowed to be negative \cite{Bar-Yossef04}. 
Here we overcome these lower bounds for smaller $s$ and indeed
will develop algorithms and lower bounds for estimating
$\sum_{i=1}^n f(\sum_{t=1}^s a_{ti})$, for a general class of functions $f:{\bf R}_+\rightarrow {\bf R}_+$.
 
We then extend these results to the following more general problem:
there is a collection of vectors that is
partitioned into $s$ parts - $W_1,W_2,\ldots ,W_s$ - and server $t$ holds $W_t$.
For each $t$ and each $i\in W_t$, there is an $n$-dimensional vector $v_i=(v_{i1},v_{i2},\ldots ,v_{in})$ wholly residing on server $t$.
Let $f:{\bf R}_+\rightarrow {\bf R}_+$
and $g:{\bf R}_+^k\rightarrow {\bf R}_+$ be 
functions. For a natural number $k$, define the $k$-th generalized moment $M(f,g,k)$ as
$$M(f,g,k) = \sum_{j_1,j_2,\ldots ,j_k\in [n] \, \text{distinct}} f\left(\sum_{i} g(v_{i,j_1},v_{i,j_2},\ldots ,v_{i,j_k})\right).$$

There are many applications of higher-order correlations, and we only mention several here. 
For a document collection, we seek statistics (second, third and higher moments) of the number of 
documents in which each trigram (triples of terms) occurs.
For a bipartite graph $G(V_1,V_2,E)$ and constants $(r,u)$, we want to estimate the number of $K_{r,u}$ 
(complete bipartite graph) subgraphs.
For a time series of many events, we want to estimate the number of tuples $(E_1,E_2,\ldots ,E_r; t_1,t_2,\ldots ,t_u)$ 
for which each of the events $E_1,E_2,\ldots ,E_r$ occurs at each of the times $t_1,t_2,\ldots ,t_u$. 

Conceptually, for each $i$, we can think of a vector $a_i$ with ${n\choose k}$ components -
one for each distinct tuple $(j_1,j_2,\ldots ,j_k)$. Suppose $a_{i; j_1,j_2,\ldots ,j_k}=g(v_{i,j_1},v_{i,j_2},\ldots ,v_{i,j_k}),$
and let $a_t=\sum_{i\in W_t} a_i.$
Our first theorem describes a way of estimating $M(f,g,k)$ up to a $(1+\eps)$-factor, where, each server uses
polynomial time and polynomial space, but we try to optimize total communication while keeping the number of rounds constant.
For this algorithm, server $t$ explicitly constructs the vector $a_t$ first, so it uses $O(n^k|W_t|)$ space.
Thereafter the space is linear in the total size of all the $a_t$.
Our second theorem shows how to reduce space to linear in $n$. This algorithm does not construct $a_t$ explicitly, but
instead performs a rejection sampling procedure. 

Before stating our theorems, we need some notation. Let $c_{f,s}$ be the least positive real number such that
\begin{equation}\label{cfs}
f(x_1+x_2+\cdots +x_s)\leq c_{f,s} (f(x_1)+f(x_2)+\cdots +f(x_s))\; \forall x_1,x_2,\ldots ,x_s\in {\bf R}_+.
\end{equation}
Note that for $f(x)=x^k$ (as in the $k$-th frequency moment), $c_{f,s}=s^{k-1}$, since for any non-negative real numbers
$b_1,b_2,\ldots ,b_s$, we have $(b_1+b_2+\cdots +b_s)^k\leq s^{k-1}(b_1^k+b_2^k+\cdots +b_s^k),$ and taking
$b_t=1$, we see that the factor $s^{k-1}$ cannot be improved.

\paragraph{Model.}
The communication and computations are not assumed to be synchronous. We arbitrarily denote one of the $s$ servers as the Central Processor (CP).
A round consists of the CP sending a message to each server and each server sending an arbitrary length message to the CP. A round
is complete when the CP has received messages from all servers from which it is expecting a message in that round.
All servers communicate only with the CP, which, up to a factor of two, is equivalent to the servers communicating directly with
each other (provided they indicate in their message who the message is being sent to). For formal details of this model, we refer the reader
to Section 3 of \cite{beopv13}. 
Our algorithms take polynomial time, linear space and $O(1)$ rounds of communication. 
%
\paragraph{Our Results}
\paragraph{Low-rank matrix approximation and approximate PCA.}
Our first set of results is for low-rank approximation: given an $n\times d$ matrix $A$, a positive integer $k$ and
$\varepsilon>0$, find an $n\times d$
matrix $B$ of rank at most $k$ such that
\[
||A-B||_F\leq (1+\varepsilon) \cdot  \min_{X: \text{rank}(X)\leq k} ||A-X||_F.
\]
Here, for a matrix $A$, the Frobenius norm $||A||_F^2$ is the sum of squares of the entries of $A$. A basis for the rowspace of $B$
provides an approximate $k$-dimensional subspace to project the rows of $A$ onto, and so is a form of approximate PCA. 
We focus on the frequently occurring case when $A$ is rectangular, that is, $n \gg d$.
\begin{theorem}\label{thm:low-rank-arbitrary}
Consider the arbitrary partition model where an $n\times d$ matrix $A_t$ resides in server $t$ and the data matrix
$A=A^1+A^2+\cdots +A^s$.
For any $1 \ge \eps > 0$, there is an algorithm that, on termination, leaves a
$n\times d$ matrix $C^t$ in server $t$ such that the matrix $C=C^1+C^2+\cdots +C^s$ is of rank $k$ and with arbitrarily large constant probability achieves
$
\|A-C\|_F \le (1+\eps) \min_{X: \text {rank}(X)\leq k} ||A-X||_F,$
using linear space, polynomial time and with total communication complexity $O(sdk/\eps + sk^2/\eps^4)$ real numbers.
Moreover, if the entries of each $A_t$
are $b$ bits each, then the total communication is $O(sdk/\eps + sk^2/\eps^4)$ words each consisting of $O(b + \log(nd))$ bits.
\end{theorem}
In contrast to the guarantees in Theorem \ref{thm:low-rank-arbitrary}, in the streaming model even with multiple passes, a simple encoding argument formalized in Theorem 4.14 of \cite{CW09} shows the problem requires $\Omega(n+d)$ communication. We bypass this problem by allowing the $s$ different servers to locally output a matrix $C_t$ so that $\sum_t C_t$ is a $(1+\eps)$-approximation to the best rank-$k$ approximation. We are not aware
of any previous algorithms with less than $n$ communication in the arbitrary partition model. 

In the row-partition model, in which each row of $A$ is held by a unique server, there is an $O(sdk/\eps)$ word upper bound due to \cite{fss13}. This is also achievable by the algorithms of \cite{gp13,bkl13,bklw14}. As the row-partition model is a special case of our model in which 
for each row of $A$, there is a unique server with a non-zero vector on that row, our result implies their result up to the low 
order $O(sk^2/\eps^4)$ term, but in a stronger model. For example, 
consider the case in which a customer corresponds
to a row of $A$, and a column to his/her purchases of a specific item. These purchases could be distributed across multiple servers corresponding to 
different vendors. Or in the case of search data, each column could correspond to a search term of a user, and the searches may be distributed
across multiple servers for storage and processing considerations. These examples are captured by the arbitrary partition model but not
by the row partition model. 

The technique for our upper bound is based on a two-stage adaptive sketching process, and 
has played an important role in several followup works, 
including CUR Matrix Factorizations of \cite{bw14} and subspace embeddings for the polynomial kernel by \cite{anw14}. 

We also show an $\tilde{\Omega}(skd)$ communication lower bound,
showing our algorithm is tight up to a $\tilde{O}(1/\eps)$ factor. The argument involves an upper bound showing
how a player can communication-efficiently learn a rank-$k$ matrix given only a basis for its row space. 

\begin{theorem}\label{thm:rank-k-lower-bound}
Suppose each of $s$ servers has an $n \times d$ matrix $A^i$ and the CP wants to compute a rank-$k$ approximation of
$A = \sum_{i=1}^s A^i$ to within relative error $\eps \geq 0$. The total communication required is $\tilde{\Omega}(skd)$ bits. 
Note that the lower bound holds for computing a $(1+\eps)$-approximation for any $\eps \geq 0$. 
\end{theorem}

\paragraph{Frequency moments and higher-order correlations.} Our next set of results are for estimating higher moments and higher-order correlations of distributed data. 

\begin{theorem}\label{freq-moments-thm-1}
Let $f:{\bf R}_+\rightarrow {\bf R}_+$ and $c_{f,s}$ be as in (\ref{cfs}). There are $s$
polynomial time, linear space bounded servers, where server $t$ holds a non-negative $n$-vector
$a_t = (a_{t1},a_{t2},\ldots ,a_{tn})$. We can estimate
$\sum_{i=1}^n f\left(\sum_{t=1}^s a_{ti}\right)$
up to a $(1+\varepsilon)$ factor by an algorithm using $O(s^2c_{f,s}/\varepsilon^2)$ total words of communication 
(from all servers) in $O(1)$ rounds. 
Moreover, any estimation up to a $(1+\varepsilon)$ factor needs in the worst case $\Omega(c_{f,s}/\varepsilon )$ bits
of communication.
\end{theorem} 

We remark that the lower bound applies to {\em any} function $f$ with parameter $c_{f,s}$, not a specific family of such functions.

\begin{theorem}\label{freq-moments-thm-2}
Let $f:{\bf R}_+\rightarrow {\bf R}_+$, $g:{\bf R}_+^k\rightarrow {\bf R}_+$
be monotone functions with $c_{f,s}$ as in (\ref{cfs}).
$k\in O(1)$ is a natural number
and let $M(f,g,k)$ be the generalized moment.
We can approximate $M(f,g,k)$ to relative error $\varepsilon $ by an algorithm with communication at most
$O(s^3c_{f,s}/\varepsilon^2)$ words in $O(1)$ rounds. Further, we use polynomial time and
linear space.
\end{theorem}
A key feature of this algorithm, and our following ones, is worth noting: they involve {\it no dependence} on $n$ or $\ln n$,
so they can be used when $a_t$ are implicitly specified and $n$ itself is very large, possibly infinite
(provided, we can communicate each index $i$). In the theorem below $\Omega$ is the set of coordinates of each vector. It is analogous to
$[n]$. We use $\sum_{x\in\Omega}$, which when $\Omega$ is infinite and the probabilities are densities, should be replaced with an integral;
our theorem is also valid for the case when we have integrals.
\begin{theorem}\label{sampling}
Let $f:{\bf R}_+\rightarrow {\bf R}_+$, $g:{\bf R}_+^k\rightarrow {\bf R}_+$
be monotone functions with $c_{f,s}$ as in (\ref{cfs}).
Server $t$ is able to draw (in unit time)
a sample $x\in \Omega$ according to a probability distribution $h_t$ on $\Omega$. Also, server $t$ can
estimate $\sum_{x\in\Omega}f(h_t(x))$. Then with $O(s^3c_{f,s}/\varepsilon ^2)$ words of
communication, CP can estimate
$\sum_{x\in\Omega} f\left(\sum_{t=1}^s h_t(x)\right)$
to within relative error $\varepsilon $.
\end{theorem}
As a special case we consider the well-studied case of frequency moments.
The best previous upper bound for the $k$-th frequency moment problem in the distributed setting is by \cite{WoodruffZ12} 
who gave an algorithm that achieves $s^{k-1} \left(\frac{C\log n}{\eps}\right)^{O(k)}$
communication, so the complexity still depends, albeit mildly, on $n$. 
Theorem \ref{freq-moments-thm-1} implies an algorithm with $O(s^{k+1}/\eps^2)$ words of communication. We further improve this:
\begin{theorem}\label{freq-moments-thm-4}
There are $s$ servers, with server $t$ holding a
non-negative vector $a_t=(a_{t1},a_{t2},
\ldots ,a_{tn})$.\footnote{The vector $a_t$ need not be written down explicitly in server $t$.
it just has to have the ability to (i) find $\sum_{i=1}^na_{ti}^k$ to relative error
$\varepsilon$ and draw a sample according to $\{ a_{ti}^k/\sum_{j}a_{tj}^k\}$.}
Then,  to estimate
$A=\sum_{i=1}^n\left(\sum_{t=1}^s a_{ti}\right)^k$
to within relative error $\varepsilon$,
there is an algorithm that communicates $O((s^{k-1}+s^3)(\ln s/\varepsilon )^3)$ words
\footnote{Each communicated word is either an index $i$ or a value $a_{ti}$.} in $O(1)$ rounds.
\end{theorem}
Thus, for $k \ge 4$, the complexity is $\tilde{O}(s^{k-1}/\eps^3)$.
Our algorithm has no dependence on $n$, though it does have the restriction that $k \geq 4$. It nearly matches a known lower bound 
of $\Omega(s^{k-1}/\eps^2)$ due to \cite{WoodruffZ12}. In Theorem \ref{thm:Lip}, we extend the algorithm
and its near-optimal guarantees to a broader class of functions.
%
%

\section{Low-rank Approximation}\label{sec:low-rank}
For a matrix $A$, define $f_k(A)$ as: $f_k(A) =\min_{X: \text {rank}(X)\leq k} ||A-X||_F.$
Recall that the {\bf rank-$k$ approximation problem} is the following: Given an $n\times d$ matrix $A$, and
$\varepsilon>0$, find an $n\times d$
matrix $B$ of rank at most $k$ such that
$||A-B||_F\leq (1+\varepsilon) \cdot  f_k(A)$.

\subsection{Upper bound for low rank approximation}
One of the tools we need is a subspace embedding. 
A random $m \times n$ matrix $P$ with $m = O(d/\eps^2)$ is a {\it subspace embedding} if for all vectors
$x \in \mathbb{R}^d$, $\|PAx\|_2 = (1\pm \eps)\|Ax\|_2$. There are many choices for $P$,
including a matrix of i.i.d. $N(0, 1/m)$
random variables or a matrix of i.i.d. Rademacher random variables (uniform in $\{-1/\sqrt{m}, +1/\sqrt{m}\}$) with $m = O(d/\eps^2)$
(combining the Johnson-Lindenstrauss transform with a standard net argument) by \cite{AV99,Ach03,AV06}. 
With a slightly larger value of $m$, 
one can also use
Fast Johnson-Lindenstrauss transforms by \cite{AC09} and the many optimizations to them,
or the recent fast sparse subspace embeddings
by \cite{CW13} and its optimizations in \cite{nn12,mm13}. Such mappings can also be composed with
each other. 

We are mainly
concerned with communication, so we omit the tradeoffs of different
compositions
and just use a composition for which $m = O(d/\eps^2)$, $PA$ is an $m \times d$ matrix of words each consisting of
$O(b + \log(nd))$ bits, and $P$ can be specified using $O(d \log n)$ bits (using a $d$-wise independent hash function, as
first shown in \cite{CW09}), see Theorem \ref{thm:kv} below.
Since we will assume that $b$ is at least $\log n$, the $O(d \log n)$ bits to specify $P$ will be negligible,
though we remark that the number of bits to specify $P$ can be further reduced using results of 
\cite{kmn11}.

We will prove the following property about the top $k$ right singular vectors of $PA$ for a subspace
embedding $P$. 
\begin{theorem}\label{thm:kv}
Suppose $A$ is an $n \times d$ matrix. Let $P$ be an $m \times d$ matrix for which
$(1-\eps)\|Ax\|_2 \leq \|PAx\|_2 \leq (1+\eps)\|Ax\|_2$ for all $x \in \mathbb{R}^d$, that is,
$P$ is a subspace embedding for the column space of $A$. Suppose
$VV^T$ is a $d \times d$ matrix which projects vectors in $\mathbb{R}^d$ onto the space of
the top $k$
singular vectors of $PA$. Then
$\|A-AVV^T\|_F \leq (1+O(\eps))\cdot f_k(A).$
Furthermore, if $m =  O(d/\eps^2)$ and
$P$ is a random sign matrix with entries uniform in $\{-1/\sqrt{m}, 1/\sqrt{m}\}$,
then with $O(d)$-wise independent entries,
$P$ satisfies the above properties with probability at least\footnote{$\exp(-d)$ denotes $2^{-\Theta(d)}$.} 
$1-\exp(-d)$.
\end{theorem}
We will combine this property with the following known property. 
\begin{theorem}\label{thm:sketch}(combining Theorem 4.2 and the second part of Lemma 4.3 of \cite{CW09})
Let $S \in \mathbb{R}^{m \times n}$ be a random sign matrix with $m = O(k \log (1/\delta)/\eps)$ in which
the entries are $O(k + \log(1/\delta))$-wise independent.
Then with probability at least $1-\delta$, if $UU^T$ is the $d \times d$ projection matrix onto the row space of $SA$,
then if $(AU)_k$ is the best rank-$k$ approximation to matrix $AU$, we have
$$\|(AU)_kU^T-A\|_F \leq (1+O(\eps))\|A-A_k\|_F.$$
\end{theorem}

We can now state the algorithm, which we call {\sc AdaptiveCompress}. 
%
\begin{figure*}[h]
\begin{center}
\fbox{
\parbox{\textwidth}{
\medskip
{\sc AdaptiveCompress($k$,$\eps$, $\delta$)}
\begin{enumerate}
\item Server $1$ chooses a random seed for an $m \times n$ sketching matrix $S$ as in Theorem \ref{thm:sketch},
given parameters $k, \eps,$ and $\delta$, where $\delta$ is a small positive constant.
It communicates the seed to the other servers.
\item Server $i$ uses the random seed to compute $S$, and then $SA^i$, and sends it to Server $1$.
\item Server $1$ computes $\sum_{i=1}^s SA^i = SA$. It computes an $m \times d$
orthonormal basis $U^T$ for the row space of $SA$, and sends $U$ to all the servers.
\item Each server $i$ computes $A^i U$. 
\item Server 1 chooses another random seed for a $O(k/\varepsilon^2)\times n$ matrix $P$ 
which is to be $O(k)$-wise independent and communicates this seed to all servers.
\item The servers then agree on a subspace embedding matrix $P$ of Theorem \ref{thm:kv} for $AU$, 
where $P$ is an $O(k/\eps^3) \times n$ matrix which can be described with $O(k \log n)$ bits.
\item Server $t$ computes $PA_tU$ and send it to Server $1$.
\item Server $1$ computes $\sum_{t=1}^s PA_tU = PAU$. It computes $VV^T$, which is an $O(k/\eps) \times O(k/\eps)$ projection matrix onto the top $k$ singular vectors of $PAU$, and sends $V$ to all the servers.
\item Server $t$ outputs $C_t = A_t UVV^TU^T$. Let $C=\sum_{t=1}^s C_t$. $C$ is not computed explicitly.
\end{enumerate}
}
}
\end{center}
\end{figure*}
In {\sc AdaptiveCompress}, the matrix $P$ is of size $O(k/\eps^3) \times n$.

\begin{proof} (of Theorem \ref{thm:kv}.)
Suppose $P$ is a subspace embedding for the column space of $A$.
Form an orthonormal basis of ${\bf R}^d$ using the right singular vectors of $PA$. Let
$v_1,v_2,\ldots ,v_d$ be the basis.
\begin{align*}
||A-A\sum_{i=1}^kv_iv_i^T||_F^2&=\sum_{i=k+1}^d |Av_i|^2\leq (1+\varepsilon)^2\sum_{i=k+1}^d|PAv_i|^2 = (1+\eps)^2f_k^2(PA).
\end{align*}
Also, suppose now $u_1,u_2,\ldots ,u_d$ is an orthonormal basis consisting of the singular vectors
of $A$. Then, we have
\begin{align*}
f_k(PA)^2&\leq ||PA-PA\sum_{i=1}^ku_iu_i^T||_F^2 = \sum_{i=k+1}^d |PAu_i|^2 
\leq (1+\varepsilon )^2\sum_{i=k+1}^d|Au_i|^2 = (1+\eps)^2f_k(A)^2.
\end{align*}
Thus,$
\|A - A \sum_{i=1}^k v_i v_i^T\|_F^2 \le (1+\eps)^4 f_k(A)^2,$
as desired.

For the second part of the theorem, regarding the choice of $P$, fix
attention on one particular $x\in {\bf R}^d$.
We apply Theorem 2.2 of \cite{CW09}
with $A,B$ of that theorem both set to $Ax$ of the current theorem
and $m = O(d/\eps^2)$ in the notation of that theorem. This states that for $m = O(d/\eps^2)$,
if $P$ is an $m \times n$
matrix with $O(d)$-wise independent entries uniform in $\{-1/\sqrt{m}, +1/\sqrt{m}\}$, then for any fixed
vector $x$, $\|PAx\|_2 = (1 \pm \eps)\|Ax\|_2$ with probability $1-\exp(-d)$. We combine this with
Lemma 4 in Appendix A of \cite{ahk06}, based on \cite{fo05},
to conclude that for all vectors $x$, $\|PAx\|_2 = (1 \pm \eps)\|Ax\|_2$ with
probability $1-\exp(-d)$ (for a different constant in the $\exp()$ function).
\end{proof}

\begin{proofof}{of Theorem \ref{thm:low-rank-arbitrary}}
By definition of the {\sc AdaptiveCompress} protocol, we have $\|A-C\| = \|A- AU V V^T U^T\|,$ where
all norms in this proof are the Frobenius norm.

Notice
that $UU^T$ and $I_d-UU^T$ are projections onto orthogonal subspaces, where $I_d$ is the $d \times d$ identity
matrix. It follows by the Pythagorean theorem applied to each row that
\begin{eqnarray}\label{eqn:first}
\|AU V V^T U^T - A\|^2
& = & \|(AU V V^T U^T -A)(U U^T)\|^2 + \|(AU V V^T U^T-A)(I - UU^T)\|^2\nonumber \\
& = & \|AUVV^TU^T - AUU^T\|^2 + \|A-AUU^T\|^2,
\end{eqnarray}
where the second equality uses that $U^TU = I_c$, where $c$ is the number of columns of $U$.

Observe that the row spaces of $AUVV^TU^T$ and $AUU^T$ are both in the row space of $U^T$, and
therefore in the column space of $U$. It follows that since $U$ has orthonormal columns,
$\|AUVV^TU^T-AUU^T\| = \|(AUVV^TU^T-AUU^T)U\|,$ 
and therefore
\begin{eqnarray}\label{eqn:second}
\|AUVV^TU^T-AUU^T\|^2 + \|A-AUU^T\|^2 & = & \|(AUVV^TU^T-AUU^T)U\|^2 + \|A-AUU^T\|^2 \nonumber \\
& = & \|AUVV^T - AU\|^2 + \|A-AUU^T\|^2,
\end{eqnarray}
where the second equality uses that $U^T U = I_c$. Let $(AU)_k$ be the best rank-$k$ approximation to
the matrix $AU$. By Theorem \ref{thm:kv}, with probability $1-o(1)$,
$\|AUVV^T - AU\|^2 \leq (1+O(\eps)) \|(AU)_k - AU\|^2,$
and so
\begin{eqnarray}\label{eqn:third}
\|AUVV^T-AU\|^2 + \|A-AUU^T\|^2
& \leq & (1+O(\eps)) \|(AU)_k - AU\|_2^2 +  \|A-AUU^T\|^2 \nonumber \\
& \leq & (1+O(\eps)) ( \|(AU)_k - AU\|_2^2 +  \|A-AUU^T\|^2).
\end{eqnarray}
Notice that the row space of $(AU)_k$ is spanned by the top $k$ right singular
vectors of $AU$, which are in the row space of $U$. Let us write $(AU)_k = B \cdot U$, where $B$ is a rank-$k$ matrix.

For any vector $v \in \mathbb{R}^d$
, $vUU^T$ is in the rowspace of $U^T$, and since the columns of $U$ are orthonormal, $\|vUU^T\|^2 = \|vUU^TU\|^2 = \|vU\|^2$,
and so 
\begin{eqnarray}\label{eqn:fourth}
\|(AU)_k - AU\|^2 + \|A-AUU^T\|^2& = & \|(B -A)U\|^2 + \|A(I-UU^T)\|^2 \nonumber \\
& = & \|BUU^T - AUU^T\|^2 + \ \|AUU^T - A\|^2.
\end{eqnarray}
We apply the Pythagorean theorem to each row in the expression in (\ref{eqn:fourth}), noting that
the vectors $(B_i-A_i)UU^T$ and $A_iUU^T-A_i$ are orthogonal, where $B_i$ and $A_i$ are the $i$-th rows of $B$ and $A$,
respectively. Hence,
\begin{eqnarray}\label{eqn:fifth}
\|BUU^T - AUU^T\|^2 + \|AUU^T - A\|^2
& = & \|BUU^T - A\|^2 = \|(AU)_k U^T - A\|^2,
\end{eqnarray}
where the first equality uses that
$$\|BUU^T-A\|^2 = \|(BUU^T-A)UU^T\|^2 + \|(BUU^T-A)(I-UU^T)\|^2 = \|BUU^T - AUU^T\|^2 + \|AUU^T - A\|^2,$$ and 
the last equality uses the definition of $B$. By Theorem \ref{thm:sketch},
with constant probability arbitrarily close to $1$, we have
\begin{eqnarray}\label{eqn:sixth}
\|(AU)_k U^T - A\|^2 & \leq & (1+O(\eps)) \|A_k - A\|^2.
\end{eqnarray}
It follows by combining (\ref{eqn:first}), (\ref{eqn:second}), (\ref{eqn:third}), (\ref{eqn:fourth}), (\ref{eqn:fifth}),
(\ref{eqn:sixth}),
that $\|AU V V^T U^T - A\|^2 \leq (1+O(\eps))\|A_k-A\|^2$, which shows the correctness property of {\sc AdaptiveCompress}.

We now bound the communication. In the first step, by Theorem \ref{thm:sketch}, $m$ can be set to $O(k/\eps)$
and the matrix $S$ can be described using a random seed that is $O(k)$-wise independent. The communication of
steps 1-3 is thus $O(sdk / \eps)$ words. By Theorem \ref{thm:kv}, the remaining steps take $O(s(k/\eps)^2/\eps^2) =
O(sk^2/\eps^4)$ words of communication.

To obtain communication with $O(b + \log(nd))$-bit words if the entries of the matrices $A_t$ are specified by $b$ bits, Server 1 can
instead send $SA$ to each of the servers. The $t$-th server then computes $PA_t (SA)^T$ and sends this to Server 1.
Let $SA = RU^T$, where $U^T$ is an orthonormal
basis for the row space of $SA$, and $R$ is an $O(k/\eps) \times O(k/\eps)$ change of basis matrix.
Server 1 computes $\sum_t PA_t (SA)^T = PA(SA)^T$ and sends this to each of the servers. Then, since each of the servers
knows $R$, it can compute $PA(SA)^T (R^T)^{-1} = PAU$. It can then compute the SVD of this matrix, from which it obtains
$VV^T$, the projection onto its top $k$ right singular vectors. Then, since Server $t$ knows $A_t$ and $U$, it can compute
$A_t U(VV^T)U^T$, as desired. Notice that in this variant of the algorithm what is sent is $SA_t$ and $PA_t(SA)^T$, which
each can be specified with $O(b + \log(nd))$-bit words if the entries of the $A_t$ are specified by $b$ bits.
\end{proofof}

\subsection{Lower bound for low-rank approximation}

Our reduction is from the multiplayer SUM problem.
\begin{theorem}\label{thm:multi-lb}(\cite{PhilipsVZ12})
 Suppose each of $s$ players has a binary vector $a^i$ with $n$ bits and the first player wants to compute $\sum_{i=1}^s a^i$ mod $2$
with constant probability.  
Then the total communication needed is $\Omega(sn)$ bits.
\end{theorem}

\begin{proof}(of Theorem \ref{thm:rank-k-lower-bound}.)
We reduce from the $s-2$ player SUM problem, in which each player has a $k \times d$ binary matrix $A^i$ and the first player wants to learn their sum. By Theorem \ref{thm:multi-lb}, this problem needs $\Omega(skd)$ communication, since even the mod $2$ version of the problem requires this amount of communication. Now consider the $s$-player problem $s$-RESTRICT-SUM in which the first player has $I_{d}$, the second player has $-I_{d}$, the remaining $s-2$ players have a $k \times d$ binary matrix $A^i$ and the first player wants to learn the sum of all inputs. This also requires $\Omega(skd)$ communication. This follows since if this problem could be solved with $o(skd)$ communication, then SUM with $s-2$ players would have $o(skd)$ communication by a simulation in which the first player of the $(s-2)$-SUM problem simulates the first three players of the $s$-RESTRICT-SUM problem.

In our $s$-player low-rank approximation problem, we give the first player $I_{d}$, the second player $-I_{d}$, and the remaining $s-2$ players each has a random $k \times d$ binary matrix $A^i$. Note that there is a unique rank-$k$ approximation to the sum of the $k$ player inputs, namely, it is the matrix $\sum_{i=3}^s A^i$. It follows that any algorithm which outputs a projection matrix $VV^T$ for which $\|A-AVV^T\|_F^2 \leq (1+\eps) \min_{X: \textrm{ rank}(X) \leq k}$, for any $\eps \geq 0$, must be such that $VV^T$ is a projection onto the row space of $\sum_{i=3}^s A^i$. This follows because $(1+\eps) \min_{X: \textrm{ rank}(X) \leq k} \|A-X\|_F = (1+\eps) \cdot 0 = 0$. 

Now, since the first player has $I_{d}$, his output is $I_{d} VV^T$, where the row space of $V^T$ equals the row space of $\sum_{i=3}^s A^i$. Suppose the total communication of our problem is $C$.

We use this to build a protocol for $s$-RESTRICT-SUM, which has the same inputs as in our $s$-player low rank approximation problem. Notice that $A =\sum_{i=3}^s A^i$ is a $k \times d$ matrix with rows in $\{0, 1, 2, ..., s-2\}^d$.

\noindent
{\bf Claim.} The span of the rows of $A$ can intersect $\{0,1,\ldots, s-2\}^d$ in at most $(2s)^k$ distinct points.

\begin{proof}
Let ${\mbox rowspace}(A)$ denote the row space of $A$.
We will bound the size of ${\mbox rowspace}(A) \cap GF(p)^d$ for prime $p$ with $s-2 < p < 2(s-2)$, where $GF(p)$ is the finite field with elements $\{0, 1, 2, \ldots, p-1\}$, and $GF(p)^d$ is the vector space over $GF(p)$. This will be an upper bound
on the size of ${\mbox rowspace}(A) \cap \{0,1,\ldots, s-2\}^d$.
Since ${\mbox rowspace}(A)$ is $k$-dimensional, so is ${\mbox rowspace}(A) \cap GF(p)^d$. Hence the intersection has at most $k$ linearly independent points. These $k$ linearly independent points can be used to generate the remaining points in ${\mbox rowspace}(A) \cap GF(p)^d$.
The number of distinct combinations of these points is at most $ p^k < (2s)^k$, bounding the intersection size.
\end{proof}

Next, players $3,\ldots, s$ agree on random $\{+1,-1\}^d$ vectors $u^1, \ldots u^{k'}$ where $k' =k\log 4s$ via a public coin. The entries of $u^1, ..., u^{k'}$ need only be $O(k \log s)$-wise independent, and as such can be agreed upon by all the players using only $O(sk\log s)$ bits of communication. Each player $i$ then computes the inner products $A^i_j \cdot u^1$, \ldots, $A^i_j \cdot u^{k'}$ for each $j \in \{1,\ldots, k\}$. Here $A^i_ j$ denotes the $j$'th row of a the $i$'th player's matrix $A^i$.

The players $P_3, ..., P_s$ send all of these inner products to $P_1$. The latter, for all rows $j \in \{1, ..., k\}$, computes the inner products $A_j \cdot u^1 , \ldots, A_j\cdot u^{k \log s}$, where $A = \sum_{i=3}^s A^i$. This can be done using $\tilde{O}(sk^2)$ communication. Since $P_1$ now has $V^T$, he can compute the $O(s)^k$ points in $\{0, 1, ..., s-2\}^d$ that each row of $A$ could possibly be. Let $p$ be one such possible point. For each row $j$, $P_1$ checks if $A_j\cdot u^l$ = $p \cdot u^l$
for every $l \in \{1, 2, \ldots, k'\}$. He decides $p=A_j$ iff all $k\log s$ equalities hold for $A_j$. In this way, he can reconstruct $A =\sum_{i=3}^s A_i$. The number $k \log s$ of the different $u$ vectors is chosen so that by a union bound, the procedure succeeds with high probability. 

We can thus solve the $s$-RESTRICT-SUM problem using our $s$-player low rank problem with communication $C + \tilde{O}(sk^2)$, where $C$ was the communication of our low-rank problem. Therefore,
$C + \tilde{O}(sk^2) = \Omega(skd)$, which implies $C = \tilde{\Omega}(skd)$ since $k < d$.
\end{proof}

\section{Frequency Moments and Higher Order Correlations}
In this section, we prove Theorems (\ref{freq-moments-thm-1}), (\ref{freq-moments-thm-2}) and (\ref{freq-moments-thm-4}).

We begin with some common notation. 
For $t\in [s], i\in [n]$:
\[
C_t=\sum_{i=1}^n f(a_{ti})\; ;\;  B_i=\sum_{t=1}^s f(a_{ti}) \; ;\; 
A_i = f\left( \sum_{t=1}^s a_{ti} \right).
\]
Let $B=\sum_iB_i=\sum_tC_t\; ;\;  A= \sum_i A_i$.
The task is to estimate $A$. We analyze the following algorithm.
Let $l=100\frac{s \cdot c_{f,s}}{\varepsilon ^2}$. The parameters in the algorithm will be specified presently.

\begin{figure*}[h]
\begin{center}
\fbox{
\parbox{\textwidth}{
\medskip
{\sc DistributedSum($\eps$)}
\begin{enumerate}
\item For $t\in [s]$, server $t$ computes $C_t$ and all servers send their $C_t$ to CP. This is round 1.
\item CP does $l$ i.i.d. trials, in each picking a $t$, with probabilities $\{ C_t/B\} $. Let $d_t$ be the
number of times it picks $t$. CP sends $d_t$ to server $t$.
\item Server $t$ picks $d_t$ samples $i_1,i_2,\ldots $ in i.i.d. trials, each according to probabilities
$\{ f(a_{ti})/C_t\}$ and sends the $d_t$ indices to CP. Round 2 is complete when CP receives all these indices.
\item CP collects all the samples. Let $S$ be the set of sampled $i$ (so, $|S|=l$). CP sends all of $S$ to all
servers.
\item Server $t$ sends $a_{ti}$ for all $i\in S$ to CP.
\item CP computes $A_i,B_i$ for all $i\in S$ and outputs $\frac{B}{l} \sum_{i\in S} \frac{A_i}{B_i}$ as its
estimate of $A$.
\end{enumerate}
}
}
\end{center}
\end{figure*}

\begin{proof} (of Theorem (\ref{freq-moments-thm-1})):

To analyze the algorithm, we think of it differently: suppose CP picks $t$ for the first of its $l$ trials and
asks that $t$ to pick $i$ according to its $\{ f(a_{ti}/C_t\}$. Let $X$ be the random variable $BA_i/B_i$ for that $i$.
Clearly the estimate made by the algorithm can be viewed as the average of $l$ i.i.d. copies of $X$. So it will
suffice to show that (i) $X$ is unbiased : I.e., $E(X)=A$ and (ii) Var$(X)\leq c_{f,s}sA^2$ (whence,
the variance of the average of $l$ i.i.d. copies of $X$ would have variance at most $\varepsilon^2A^2$ giving us the
relative error bound.)

The first part is
easy: Let $p_i$ be the probability that we pick $i$ by this process. Clearly,
$p_i=\sum_{t=1}^s\prob (\text{ CP picked $t$ })\prob ( t\text{ picks }i)=\sum_t \frac{C_t}{B}\frac{f(a_{ti})}{C_t}=\frac{B_i}{B}.$
So, $E(X)=\sum_{i=1}^n B\frac{A_i}{B_i}\frac{B_i}{B}=A,$ proving (i).
For (ii), we have
$E(X^2) = B^2\sum_i p_i \frac{A_i^2}{B_i^2}=B\sum_i \frac{A_i^2}{B_i}\leq ABc_{f,s}\leq c_{f,s}sA^2,$
since, $A_i=f(\sum_t a_{ti})\leq c_{f,s} \sum_{t=1}^s f(a_{ti})$ by the definition of $c_{f,s}$ and
by monotonicity of $f$, we have $B_i=\sum_t f(a_{ti})\leq sf(\sum_t a_{ti})$.

To prove the claimed resource bounds, note that polynomial time and linear space bounds are
obvious, since, all that each server has to do is to compute all $f(a_{ti})$, sum them up
and sample at most $l$ times. The communication is dominated by each of $s$ servers sending
$\{ a_{ti},i\in S\}$ to CP which is $sc_{f,s}/\varepsilon^2$ words per server giving us a total
of $O(s^2c_{f,s}/\varepsilon^2)$.

Now for the lower bound, we use (rather unsurprisingly) the set-disjointness problem. 
It is known (~\cite{AMS99,Bar-Yossef04,Chakrabarti03,g09,j09,WoodruffZ12})
that the following problem needs $\Omega(n)$ bits of communication even for a randomized algorithm:
we distinguish between two situations: (a) Each of $s$ servers holds a subset of $[n]$ and the subsets are
pairwise disjoint and (b) There is exactly one element common to all $s$ sets.
We reduce this problem to ours. Let $S_t$ be the subset held by server $t$.
By definition of $c_{f,s}$, there exist $x_1,x_2,\ldots ,x_s\in {\bf R}_+$ such that
$f(x_1+x_2+\cdots +x_s)=c_{f,s}(f(x_1)+f(x_2)+\cdots +f(x_s)).$
Let $n=\frac{c_{f,s}-1}{\varepsilon}.$
Let $a_{ti}$ be defined by: $a_{ti}= x_t \text{ if } i\in S_t$ and $a_{ti} = 0$ otherwise.
If the sets are disjoint, then
$\sum_{i=1}^n f\left( \sum_{t=1}^s a_{ti}\right) = \sum_{t=1}^s |S_t|f(x_t).$
In the case (b) when the sets all share one element in common,
$\sum_{i=1}^n f\left( \sum_{t=1}^sa_{ti}\right) = \sum_{t=1}^s (|S_t|-1)f(x_t) + f(x_1+x_2+\cdots +x_s)
= \sum_{t=1}^s |S_t|f(x_t)+ (c_{f,s}-1) \sum_t f(x_t)=\sum_{t=1}^s |S_t|f(x_t)+ \varepsilon n\sum_t f(x_t).$
Since $|S_t|\leq n$, it follows that if we can estimate $\sum_i f(\sum_t a_{ti})$ to relative error $\varepsilon$, then we can distinguish the
two cases. But it is known that this requires $\Omega (n)$ bits of communication which is $\Omega (c_{f,s}/\varepsilon)$
proving the lower bound.   

\end{proof}

\begin{proof} (of Theorem (\ref{freq-moments-thm-2}):
%
The only change is in the sampling algorithm:
\begin{itemize}
\item Order the $j=(j_1,j_2,\ldots ,j_k)$ lexicographically. Start with the first $j$ as the sample
and compute $a_{tj}$ by making a pass through the entire data:
For each $i\in W_t$, after $v_{i,j_1}, v_{i,j_2}, \ldots ,v_{i,j_k}$ are read, compute
                $g(v_{i,j_1},v_{i,j_2},\ldots ,v_{i,j_k})$ and sum over all $i\in W_t$.
\item Process the next $j$ similarly. After processing a $j$, say, $j=j_0$, compute $f(a_{tj_0})$ and keep a running total of
    $f(a_{tj})$ for all $j$ seen so far. Reject the old sample and replace it by the current $j_0$ with probability
            $\frac{f(a_{tj_0})}{\text{ Total of all } f(a_{tj}) \text{ including }j_0}.$
\item If the old sampled $j$ is not rejected, just keep it as the sample and go to next $j$.
\end{itemize}
The proof of correctness and linear space bound follow straightforwardly by plugging in this sampling algorithm into
Theorem (\ref{freq-moments-thm-1}). 
\end{proof}

We next turn to a more refined algorithm for estimating frequency moments with near-optimal communication, using the specific function $f(x)=x^k$. Here is the algorithm.

\begin{figure*}[h]
\begin{center}
\fbox{
\parbox{\textwidth}{
\medskip
{\sc FrequencyMoments($k$, $\eps$)}
\begin{enumerate}
    \item Pick an i.i.d. sample $S_0$ of $m=s^{k-2}/\eps^3$ indices $i$, where each $i$ is picked
        according to $\{ B_i/B\}$.
    \item Find average $\rho_i$ among the picked $i$ by getting all the $a_{ti}$
        and estimate $A$ by $\tilde A=B$ times this average.
    \item If $\tilde A\geq sB$, then declare it to be the final estimate of $A$ and stop.
    \item CP now gets an i.i.d. sample $S$ of $O(s^{k-1}(\ln s)^2/\eps^3)$ $i$ 's, each according to $\{ B_i/B\}$.
    \item For each $\beta \in \{ s^{k-1},e^{-\varepsilon }s^{k-1},e^{-2\varepsilon }s^{k-1},\ldots ,1\}$,
        CP does the following:
            \begin{enumerate}
                \item Pick a subset $T$ of $S$ of cardinality $\Omega(\beta (\ln s)^2/\eps^3)$ u.a.r.
                \item For each $i\in T$, pick a set $L$ of $l=\frac{s^{k-1}}{\beta}$ $t\in [s]$ u.a.r. Find all the $a_{ti},t\in L$ and find
                    $\frac{s^k}{l^k}\left( \sum_{t\in L}a_{ti}\right)^k$. Repeat this $\Omega(k\ln s+\ln(1/\eps))$ times and take the median of all values found to be the estimate $\tilde A_i$ of $A_i$.
                \item For each $i\in T$, take  $\tilde B_i=a_{t(i),i}$, where, $t(i)$ is defined in (\ref{tofi}).
                \item For every $i\in T$ with $\tilde A_i/\tilde B_i\in [\beta e^{-\varepsilon },\beta)$, do an exact computation
                of $A_i,B_i$ by asking every server for all the $a_{ti}$ values.
                \item From the above estimate $|S_\beta \cap T$ and compute 
$\tilde{s}_\beta = |S_\beta \cap T| |T|/|S|$ as the estimate of $|S_\beta|$.
            \end{enumerate}
    \item Return $B\sum_\beta \tilde s_\beta \beta$ as the estimate of $A$.
\end{enumerate}
}
}
\end{center}
\end{figure*}

\begin{proof} (of Theorem (\ref{freq-moments-thm-4}):
Let $B_i=\sum_{t=1}^sa_{ti}^k$ and $A_i=\left( \sum_{t=1}^s a_{ti}\right)^k$
and $\rho_i=A_i/B_i$. Note: $1\leq \rho_i\leq s^{k-1}$. CP can arrange to
pick $i\in [n]$ with probabilities $\{ B_i/B\}$, where, $B=\sum_iB_i$ as we
already saw. First pick $m=s^{k-2}/\varepsilon^3 $ sample $i\in [n]$
according to $\{ B_i/B\}$. Then, CP tells all servers all these $i$ and collects
all $a_{ti}$ and thence all $A_i,B_i$. Total communication is at most $ms\leq s^{k-1}/\varepsilon^3$.

The estimator of $A$ from one $i$ is
$X=\frac{BA_i}{B_i}.$
It is easy to see that $E(X)=A$, and 
$\text{Var}(X)\leq E(X^2) = B\sum_{i=1}^n \frac{A_i^2}{B_i}\leq ABs^{k-1},$
since each $A_i/B_i\leq s^{k-1}$.
So if we estimate $A$ by
$\tilde A=$ average of $m$ i.i.d. copies of $X$, then we would get
Var$(\tilde A)\leq \varepsilon^3 sAB$.

\begin{claim}\label{AgeqBs}
With a suitable choice of constants, if $A\ge sB)$, then,
$\prob\left( |\tilde A-A|\leq\varepsilon A\; \text{and}\; \tilde A\in (1-\eps) sB\right)\geq 1-c.$
Further, since $B$ is known and $\tilde A$ is computed the condition $\tilde A\le sB$ can be checked.
Conversely, if $A \le sB$, then,
$\prob\left( \tilde A\in sB/(1-\eps))\right)\geq 1-\eps.$
\end{claim}

\begin{proof}
In the first case, we have
$\Var(\tilde A)\le \varepsilon^3 A^2,$
from which the first assertion follows using Chebychev inequality. The checkability is clear.
For the second part,
$\prob (\tilde A\geq sB/(1-\eps))\leq \prob(\tilde A>A+ \varepsilon sB))
\leq\prob \left( \tilde A\geq A+\varepsilon \sqrt{sAB}\right)
\leq \frac{\text{Var}(\tilde{A})}{\eps^2 sAB}\leq \eps.$
\end{proof}

Given the claim, after this step, the algorithm either has found that $A > sB$ and $\tilde A$ is a
good estimate of $A$ and terminated or it knows that $A\leq sB$. So assume now $A\leq sB$.
CP now collects a set $S$ of $s^{k-1}(\ln s)^2/\varepsilon^3$ sampled
$i$ 's, each i.i.d. sampled according to $\{ B_i/B\}$. [It cannot
now afford to inform all servers of all $i\in S$.]

Let $\rho_i =A_i/B_i$. Let $\beta $ range over
$\{ s^{k-1}, e^{-\varepsilon }s^{k-1},e^{-2\varepsilon }s^{k-1},\ldots 1\}$, a total of
$O(\ln s)$ values and
let
$S_\beta = \{i\in S: \rho_i\in [\beta e^{-\varepsilon} ,\beta)\}.$
Then,
$\sum_{i\in S}\rho_i\approx \sum_\beta |S_\beta|\beta.$
Since each $\rho_i\geq 1$, we have $\sum_{i\in S}\rho_i\geq |S|$.
So we need only accurate estimates of those $|S_\beta|$ with $|S_\beta|
\geq \varepsilon |S|/\beta \ln s$. For each $S_\beta$, if we pick
u.a.r. a subset $T$ of $S$ of cardinality $\Omega(\beta(\ln s)^2/\eps^3)$, then
$\frac{|S|}{|T|}|T\cap S_\beta|$ estimates $|S_\beta |$ to within $(1\pm \eps)$ for every $\beta$ satisfying
$|S_\beta|\geq \varepsilon |S|/\beta\ln s$.

For each $\beta$, pick such a random subset $T$ from $S$. We have to recognize for each
$i\in T$, whether it is in $S_\beta$. First, for each $i\in T$, we estimate
$A_i$ as follows: We pick $l=s^{k-1}/\beta$ servers $t_1,t_2,\ldots ,t_l$
u.a.r. and take $Z_i=\frac{s}{l}(a_{t_1,i}+a_{t_2,i}+\cdots a_{t_l,i})$
as our estimate of $A_i^{1/k} = \sum_{t=1}^s a_{ti}$.
If $Y$ is the r.v. based on just one random server (namely $Y=sa_{ti}$
for a u.a.r $t$), then $\E Y=A_i^{1/k}$ and
\begin{align*}
\frac{\E(Y^2)}{(\E Y)^2} &= s^2\frac{\frac{1}{s}\sum_ta_{ti}^2}{(\sum_ta_{ti})^2}\\
&\leq s^2 \frac{\left( \frac{1}{s} \sum_t a_{ti}^k\right)^{2/k}}{A_i^{2/k}}\\
&\leq \frac{s^{2-(2/k)}}{\rho_i^{2/k}}\leq \frac{e^{2\eps/k}s^{2-(2/k)}}{\beta ^{2/k}}.
\end{align*}
From this it follows by averaging over $l$ samples that
\[
\frac{\E(Z^2)}{(\E Z)^2} \le \frac{e^{2\eps/k}s^{2-(2/k)}}{l \beta ^{2/k}} = \frac{e^{2\eps/k} s^{2-(2/k)}\beta}{s^{k-1}\beta^{2/k}} = e^{2\eps/k} \left(\frac{\beta}{s^{k-1}}\right)^{1-(2/k)} \le 1+O(\frac{\eps}{k}).
\]
We do $\Omega(k\ln s + \ln(1/\eps))$ such experiments
and take the median of all of these to drive down the failure probability for a single
$i$ to less than $\eps^2/s^k$, whence, it is small by union bound for the failure of any of
the at most $s^{k-1}/\eps^2$ indices $i$ 's. Thus all the $A_i^{1/k},i\in T$ are estimated to a
factor of $(1+\epsilon)$ by this process whp. Since $k$ is a fixed constant, this also
means $(1+O(\eps))$ relative error in the estimate of $A_i$.

Next we estimate the $B_i$ to within a factor of $\tilde{O}(s)$, i.e.,
$\frac{\tilde B_i}{B_i}\in \left[ \frac{1}{10s\ln s}, 1\right]$ whp.  We will see shortly that 
such an estimate suffices. 
For each $i\in S$, define
\begin{equation}\label{tofi}
t(i)=\text{ the index of the server which picked }i\; ;\; \tilde B_i=a_{t(i),i}^k.
\end{equation}
Then 
\[
\E(\tilde{B}_i) = \frac{\sum_{t} a_{t,i}^{2k}}{\sum_{t}a_{t,i}^k} \ge \frac{1}{s}\sum_{t} a_{t,i}^k = \frac{1}{s} B_i.
\]

We observe that 
\[
\prob (a_{t(i),i}^k\leq\delta \frac{B_i}{s})\leq\delta.
\]
Let $I_\delta(i)$ be an indicator random variable of whether
$\tilde B_i\leq \delta B_i$. Then $I_\delta(i)$ are independent:
the distribution of $S, t(i)$ is not changed if
we imagine picking  $S, t(i)$ as follows: we pick $|S|$ indices $i$ in i.i.d. trials, according
to $B_i/B$. Then for each $i$ picked, independently pick a $t(i)$, where
$\prob(t(i)=t) = a_{ti}^k / B_i$. From this, the independence of $I_\delta$
is clear. Therefore, by Chernoff the the  number of $B_i$ which are
much underestimated is small. 
Fixing $\delta = 1/(10\ln s)$, for each $\beta$, the number of $i'$ for which $B_i$ is underestimated by by less than $\delta B_i/s$ is at most a $\delta$ fraction.  

We now have estimates $\tilde \rho_i$ of each $\rho_i, i\in T$. We need to determine
from this $|S_\beta|$. From the bounds on estimation errors, we have $\tilde{\rho}_i \in [e^{-2\eps/k}\rho_i , 10(s\ln s)\rho_i]$.
Therefore, we see that only
$i\in T$ with $\rho_i\geq \beta /(10s\ln s)$ may be mistaken for an $i\in S_\beta$. We have
\[
|\{ i\in S: \rho_i\geq \frac{\beta}{10s\ln s}\}| \leq \frac{\sum_S \rho_i}{\beta/(10s\ln s)}.
\]
Moreover, 
\[
\E(\sum_{i \in S} \rho_i) = |S|\E(\rho_i) = |S| \sum_{i=1}^n \frac{A_i}{B_i} \frac{B_i}{B} = |S|\frac{A}{B} \le s|S|. 
\]
Therefore, 
\[
\E(|\{ i\in S: \rho_i\geq \frac{\beta}{s}\}|)  \leq \frac{10s^2\ln s|S|}{\beta}.
\]
The subset that intersects $T$ is then 
$\{ i\in T:\rho_i\geq \beta/s\}| \le 20 \frac{|T|}{|S|} \frac{s^2\ln s|S|}{\beta} =  O(s^2\ln^3 s/\eps^3)$. Now for these $i$ 's in $T$,
we collect all $a_{ti}$ and find $A_i,B_i$ exactly. This costs us
$O(s^3(\ln s)^3/\eps^2)$ communication. Thus the overall communication is bounded by  
\[
\frac{s^{k-1}}{\eps^3}  + \frac{s^{k-1}\ln^3 s}{\eps^3} + \frac{s^3\ln^3 s}{\eps^3} = O((s^{k-1} + s^3)\ln^3 s/\eps^3).
\] This is 
$\tilde{O}(s^{k-1}/\eps^3)$ for $k \ge 4$.
%
%
%
%
%

We have given the proof already of all assertions except the number of rounds. For the number
of rounds, the most crucial point is that though the algorithm as stated requires $O((\ln s)^c)$
rounds, we can instead deal with all $\beta$ simultaneously. CP just picks the $T$ for all of them
at once and sends them accross.
Also, we just make sure that CP communicates all choices of $t$ for each $i$ all in one round.
Also, note that
the $s^{k-2}$ sampling and checking if the first $\tilde A>\Omega(sB)$ can all be done in
$O(1)$ rounds, so also the $s^{k-1}$ sampling. Then the crude estimation of $\tilde \rho_i$
can be done in one $O(1)$ rounds followed by the finer sampling in $O(1)$ rounds.
\end{proof}

We now extend the above theorem and proof for a wide class of functions satisfying a weak Lipschitz condition (and generalizing the case of moments). 

For a monotone function $f: \R_+ \rightarrow \R_+$, define 
\[
L_f = \min r :  \forall x > y > 0, \frac{f(x)}{f(y)} \le \left(\frac{x}{y}\right)^r.
\]
Alternatively, $L_f$ is the Lipschitz constant of $f$ wrt the ``distance" $d(x,y) = \log(x)-\log(y)$, i.e., 
\[
L_f = \sup \frac{d(f(x),f(y))}{d(x,y)}.
\] 
 
For the function $f(x)=x^k$, we see that $L_f = k$. 

\begin{lemma}\label{lem:Lip}
For any function $f:\R_+ \rightarrow R_+$ with $L = L_f$, 
\[
\frac{f(\sum_{t=1}^s x_t)}{\sum_{t=1}^s f(x_t)} \le \frac{\left(\sum_{t=1}^s x_t\right)^L}{\sum_{t=1}^s x_t^L}
\] 
\end{lemma}
\begin{proof}
\begin{align*}
\frac{f(\sum_{t=1}^s x_t)}{\sum_{t=1}^s f(x_t)} &= \frac{f(\sum_{t=1}^s x_t)}{\sum_{t=1}^s x_t^L}\frac{\sum_{t=1}^s x_t^L}{\sum_{t=1}^s f(x_t)}\\ 
&\le  \frac{f(\sum_{t=1}^s x_t)}{\sum_{t=1}^s x_t^L}\min_{t} \frac{x_t^L}{f(x_t)}\\
&= \min_t  \frac{f(\sum_{t=1}^s x_t)}{f(x_t)} \frac{x_t^L}{\sum_{t=1}^s x_t^L}\\
&\le \min_t  \frac{(\sum_{t=1}^s x_t)^L}{x_t^L} \frac{x_t^L}{\sum_{t=1}^s x_t^L}\\
&=  \frac{(\sum_{t=1}^s x_t)^L}{\sum_{t=1}^s x_t^L}.
\end{align*}
\end{proof}

\begin{theorem}\label{thm:Lip}
Let $f$ be any nonnegative, superlinear real function with $L = L_f \ge 4$. Suppose there are $s$ servers, with server $t$ holding a
non-negative vector $a_t=(a_{t1},a_{t2},
\ldots ,a_{tn})$
Then,  to estimate
$A=\sum_{i=1}^nf\left(\sum_{t=1}^s a_{ti}\right)$
to relative error $\varepsilon$, there is an algorithm that communicates $O(s^{L-1}(\ln s)^3/\eps^3)$ words in $O(1)$ rounds.
\end{theorem}

The algorithm is the following, essentially the same as in the moments case, with parameters defined in terms of $L_f$ for a general function $f(.)$ in place of $x^k$. 

\begin{enumerate}
    \item Pick an i.i.d. sample $S_0$ of $m=s^{L-2}/\eps^3 $ indices $i$, where each $i$ is picked
        according to $\{ B_i/B\}$.
    \item Find average $\rho_i$ among the picked $i$ by getting all the $a_{ti}$
        and estimate $A$ by $\tilde A=B$ times this average.
    \item If $\tilde A\geq sB$, then declare it to be the final estimate of $A$ and stop.
    \item CP now gets an i.i.d. sample $S$ of $O(s^{L-1}\ln^2 s/\eps^3)$ $i$ 's, each according to $\{ B_i/B\}$.
    \item For each $\beta \in \{ s^{L-1},e^{-\varepsilon }s^{L-1},e^{-2\varepsilon }s^{L-1},\ldots ,1\}$,
        CP does the following:
            \begin{enumerate}
                \item Pick a subset $T$ of $S$ of cardinality $\Omega(\beta (\ln s)^2/\eps^3)$ u.a.r.
                \item For each $i\in T$, pick a set $T'$ of $l=\frac{s^{L-1}}{\beta}$ $t\in [s]$ u.a.r. Find all the $a_{ti},t\in T'$ and find
                    $f\left( \frac{s}{l}\sum_{t\in T'}a_{ti}\right)$. Repeat this $\Omega(L\ln s+\ln(1/\eps))$ times and take the median of all values found to be the estimate $\tilde A_i$ of $A_i$.
                \item For each $i\in T$, take  $\tilde B_i=a_{t(i),i}$, where, $t(i)$ is defined in (\ref{tofi}).
                \item For every $i\in T$ with $\tilde A_i/\tilde B_i\in [\beta e^{-\varepsilon },\beta)$, do an exact computation
                of $A_i,B_i$ by asking every server for all the $a_{ti}$ values.
                \item From the above estimate $|S_\beta \cap T$ and compute 
$\tilde{s}_\beta = |S_\beta \cap T| |T|/|S|$ as the estimate of $|S_\beta|$.
            \end{enumerate}
    \item Return $B\sum_\beta \tilde s_\beta \beta$ as the estimate of $A$.
\end{enumerate}

\begin{proof}(of Thm. \ref{thm:Lip}.)
We point out the changes in the analysis from the special case of moments. 
Now we have $A_i = f(\sum_{t=1}^s a_{ti}), B_i = \sum_{t=1}^s f(a_{ti})$ and $A = \sum_{i=1}^n A_i, B = \sum_{i=1}^n B_i$. Also $\rho_i = A_i/B_i$. 
The reader will have noticed that $k$ has been replaced by $L$ in the above algorithm. 
The first phase remains the same, and at the end we either have a good approximation for $A$ or we know that $A \le sB$. 

In the next phase we estimate $f(\sum_{t} a_{ti})^{1/L}$. To do this, we first estimate $\sum_{t=1}^s a_{ti}$, then apply $f$ to this estimate. We need to analyze the error of both parts. 
For the first part, let $Y=sa_{ti}$ as before. Then $\E(Y) = \sum_{t=1}^s a_{ti}$ and since the server used to define $Y$ is chosen uniformly at random, we have 
\begin{align*}
\frac{\E(Y^2)}{\E(Y)^2} &\le s^2 \frac{\frac{1}{s}\sum_{t=1}^s a_{ti}^2}{(\sum_{i=1}^s a_{ti})^2} \\
&\le s^2 \left(\frac{\frac{1}{s} \sum_{t=1}^s a_{ti}^L}{(\sum_{i=1}^s a_{ti})^L}\right)^{2/L} \\
&\le s^{2-(2/L)} \left(\frac{f(\sum_{t=1}^s a_{ti})}{\sum_{t=1}^s f(a_{ti})}\right)^{2/L} \quad \mbox{ (using Lemma \ref{lem:Lip})}\\ 
&= \frac{s^{2-(2/L)}}{\rho_i^{2/L}}.
\end{align*}
This is then bounded by $e^{2\eps/L}$ just as before. So we get an estimate of $a_i = \sum_{t=1}^s a_{ti}$ to within multiplicative error $e^{\eps/L}$. Let $\tilde{a}_i$ be this approximation. It remains to bound $f(\tilde{a}_i)$ in terms of $f(a_i)$. For this we observe that using the definition of $L$, if $a_i \le \tilde{a}_i$, then 
\[
1 \le \frac{f(\tilde{a}_i)}{f(a_i)} \le \left(\frac{\tilde{a}_i}{a_i}\right)^L \le e^\eps. 
\]
We get a similar approximation if $a_i > \tilde{a}_i$. 

The last phase for estimating $B_i$ and putting together the estimates for all the $\rho_i$ is again the same as in the case of moments.
\end{proof}

\noindent {\bf Acknowledgements. } We are grateful to Dick Karp and Jelani Nelson for helpful discussions, as well as the Simons
Institute at Berkeley. Santosh Vempala was supported in part by NSF award CCF-1217793 and David Woodruff by the XDATA program of the Defense Advanced
Research Projects Agency (DARPA), administered through Air Force Research Laboratory contract FA8750-12-C0323.  

\bibliography{spectralbook}
\bibliographystyle{plain}

\end{document}